\documentclass[11pt]{article}
\usepackage[margin=1.1in]{geometry}
\usepackage{cite}
\usepackage{amsmath,amssymb,amsthm}
\usepackage{url}

\newtheorem{theorem}{Theorem}[section]
\newtheorem{lemma}[theorem]{Lemma}
\newtheorem{corollary}[theorem]{Corollary}
\newtheorem{proposition}[theorem]{Proposition}
\theoremstyle{remark}
\newtheorem{remark}[theorem]{Remark}

\newcommand{\R}{\mathbf{R}}
\newcommand{\T}{\mathbf{T}}
\newcommand{\W}{\mathbf{W}}
\newcommand{\B}{\mathbf{B}}
\newcommand{\I}{\mathbf{I}}
\newcommand{\J}{\mathbf{J}}
\newcommand{\PP}{\mathbf{P}}
\newcommand{\M}{\mathbf{M}}
\newcommand{\G}{\mathbf{G}}
\newcommand{\Qh}{\widehat{\mathbf{Q}}}
\newcommand{\V}{\mathbf{V}}
\newcommand{\A}{\mathbf{A}}
\newcommand{\Sb}{\mathbf{S}}
\newcommand{\Hb}{\mathbf{H}}
\newcommand{\Cb}{\mathbf{C}}
\newcommand{\Db}{\mathbf{D}}
\newcommand{\e}{\mathbf{e}}
\newcommand{\vb}{\mathbf{v}}
\newcommand{\wb}{\mathbf{w}}
\newcommand{\pb}{\mathbf{p}}
\newcommand{\qb}{\mathbf{q}}
\newcommand{\bb}{\mathbf{b}}
\newcommand{\xb}{\mathbf{x}}
\newcommand{\yb}{\mathbf{y}}
\newcommand{\ub}{\mathbf{u}}

\title{Direct Factorization of the Karhunen--Lo\`eve Transform of AR(1) Sources}

\author{Yuriy A. Reznik\\[2pt] \normalsize yreznik@mit.edu}

\begin{document}
\maketitle

\begin{abstract}
The Karhunen--Lo\`eve transform (KLT) of the stationary AR(1) source is the classical optimality benchmark of transform coding. Its frequencies are roots of the transcendental equations of Ray and Driver, and the transform has therefore long been treated as an unstructured dense matrix. Rather than compute the exact KLT, the field turned to fast approximations of it, most famously the discrete cosine transform (DCT). This paper shows that the retreat was premature. The exact AR(1) KLT admits direct recursive factorizations, at every order $N\ge2$ and every correlation coefficient $\rho\in(0,1)$: at even orders, the transform reduces to a butterfly stage, two half-order copies of itself, and orthogonal corrections generated by rank-one boundary perturbations of the covariance; at odd orders, to two copies of the KLT of the one-sided prediction residual and an arrowhead stage absorbing the center sample. The residual KLT recurses through half order as well. All correction stages are Cauchy-structured, and applying them by the fast multipole method yields exact-KLT algorithms of $O(N\log N)$ complexity---the same order as the FFT and the fast sinusoidal transforms. The factorizations follow from the covariance matrix by the Sherman--Morrison identity and the classical secular eigenvalue updates, and degenerate, as $\rho\to1$, to the known parity splittings of the DCT-II. As secondary results, the factorization constants are shown to be algebraic functions of $\rho$---in radicals for all $N\le8$---and complete exact KLT modules for $N=2,\dots,8$ are given.
\end{abstract}


\section{Introduction}

\subsection{A 1970s question, revisited}

Transform coding was built around a concrete optimization problem: find the orthogonal transform that best decorrelates, and best compacts the energy of, a correlated source. The answer was established at the outset~\cite{KramerMathews,HuangSchultheiss}: it is the Karhunen--Lo\`eve transform (KLT), the eigenvector basis of the source covariance matrix, first carried into image coding by Habibi and Wintz~\cite{HabibiWintz}. The standard source model of the field is the stationary first-order autoregressive (AR(1)) process, with correlation coefficient $\rho$ and covariance $\R_N(\rho)=[\rho^{|i-j|}]$---the Kac--Murdock--Szeg\H{o} matrix~\cite{KMS,GrenanderSzego}. For this model, the KLT eigenproblem was reduced to an analytical solution early on. Ray and Driver~\cite{RayDriver} showed that the eigenfunctions of its continuous-time counterpart, the process with exponential covariance, are pure sinusoids, with frequencies determined by a pair of transcendental tangent equations. The discrete problem inherits this structure exactly (Section~\ref{sec:eigenstructure}). The solution, however, exposed a computational difficulty. The frequency equations have no closed form for general $\rho$, and the eigenvectors are sinusoids at incommensurate frequencies. The transform therefore presented itself as an unstructured dense matrix: a numerical eigenanalysis at design time, and $N^2$ operations per data block at run time.

The resolution that shaped the following five decades was proposed by Ahmed, Natarajan, and Rao~\cite{AhmedNatarajanRao,Ahmed91}: do not compute the KLT---replace it. Their discrete cosine transform (DCT-II) is fixed, independent of $\rho$, close to the AR(1) KLT at the high correlations typical of images, and computable by fast techniques. Chen, Smith, and Fralick~\cite{ChenSmithFralick} then showed that the DCT-II matrix itself factors into a butterfly stage and two half-size transforms, a DCT-II and a DCT-IV~\cite{Wang84,Britanak}. This opened the era of direct factorizations of trigonometric transforms~\cite{Lee,Loeffler,FeigWinograd,Heideman,RaoYip,Britanak}. A complementary escape route was found by A.~K. Jain~\cite{JainFastKLT,JainSinusoidal}. His ``fast KLT'' modifies the boundary conditions of the process class, so that a fixed sinusoidal transform becomes exactly optimal---again sidestepping, rather than solving, the original problem. R.~J. Clarke completed the classical picture by identifying the limits: as $\rho\to1$ the AR(1) KLT tends to the DCT-II~\cite{ClarkeDCT}, and as $\rho\to0$ to the DST-I~\cite{ClarkeDST}. The approximation program remains active today: modern video codecs select among fixed sinusoidal transforms---the DST-VII and DCT-VIII alongside the DCT-II in versatile video coding (VVC)~\cite{ZhaoVVC,ZhangVVC}---and graph-learning methods design KLT-like transforms from data~\cite{EgilmezGBT}.

What seems to have escaped systematic attention is the original question itself: can the KLT---the exact transform, at the actual $\rho$ of the source---be factorized? The literature offers root-finding procedures for its frequencies~\cite{TorunAkansu} and a rich spectral theory of the covariance~\cite{KMS,GrenanderSzego,GrayToeplitz,Fikioris}, but treats the transform matrix as dense. We believe Ray and Driver's result played an inadvertent role here. A transcendental characteristic equation reads like an impossibility signal: it suggests that a transform whose spectrum has no closed form cannot have a fast algorithm. But this conflates two different properties. A sparse factorization is a statement about the \emph{structure} of a matrix, not about the algebraic nature of its entries. The rotation angles in such a factorization are design-time constants. They are computed once, offline, for the $\rho$ of the source class---exactly as the cosines of a fixed transform are tabulated once. Indeed, as the factorizations below make explicit, those constants are not even transcendental: they are algebraic functions of $\rho$, given in radicals for all $N\le8$. The transcendence of the AR(1) \emph{frequencies} was, in this sense, a false alarm. It appears to have biased the field toward approximating the KLT instead of factoring it.

\subsection{Contribution}
\label{sec:contribution}

This paper derives direct recursive factorizations of the exact AR(1) KLT $\W_N(\rho)$, at even and odd orders, from the covariance matrix alone, using only classical linear-algebra tools (Section~\ref{sec:preliminaries}). The factorizations reduce the transform to butterflies, half-order KLTs, and orthogonal correction stages. The correction stages are Cauchy-structured---each entry is, up to scalings, a reciprocal difference of two eigenvalues---and applying them by the fast multipole method yields exact-KLT algorithms running in $O(N\log N)$ operations (Section~\ref{sec:complexity})---the same complexity order as the FFT and the fast sinusoidal transforms. The exact optimal transform is thus not only structured but fast, in the standard sense of the field.

For even $N$, the transform splits about its center (Theorem~\ref{thm:even}):
\begin{equation}
\W_N = \PP_N\Bigl[\bigl(\Qh^{(s)}_{N/2}\W_{N/2}\bigr)\oplus\bigl(\Qh^{(a)}_{N/2}\W_{N/2}\bigr)\Bigr]\B_N .
\label{eq:even-intro}
\end{equation}
Here $\B_N$ is the sum/difference butterfly, $\W_{N/2}$ is the same KLT at half order, $\Qh^{(s)}_{N/2}$ and $\Qh^{(a)}_{N/2}$ are two orthogonal correction stages with closed-form entries, and $\PP_N$ interleaves the outputs. As $\rho\to1$,~\eqref{eq:even-intro} degenerates to the classical splitting of the DCT-II into half-size DCT-II and DCT-IV~\cite{ChenSmithFralick,Wang84}, and the recursion collapses to the decomposition underlying Chen's fast DCT (Corollary~\ref{cor:even-limit}).

For odd $N$, the retained center sample changes the picture: the half-size copies are replaced by copies of a closely related transform, and the center joins through one extra orthogonal stage (Theorem~\ref{thm:odd}):
\begin{equation}
\begin{gathered}
\W_N = \PP_N\bigl[\A_{\lfloor N/2\rfloor+1}(\V_{\lfloor N/2\rfloor}{\oplus}1)\oplus\V_{\lfloor N/2\rfloor}\bigr]\B_N ,\\
\V_{\lfloor N/2\rfloor}=\Qh^{(\delta)}_{\lfloor N/2\rfloor}\,\W_{\lfloor N/2\rfloor}.
\end{gathered}
\label{eq:odd-intro}
\end{equation}
Here the \emph{residual KLT} $\V_{\lfloor N/2\rfloor}$ is the exact KLT of an AR(1) block conditioned on a boundary sample. It is reached from the half-order KLT by one orthogonal correction $\Qh^{(\delta)}_{\lfloor N/2\rfloor}$, so~\eqref{eq:odd-intro} again recurses into half order. The stage $\A_{\lfloor N/2\rfloor+1}$ absorbs the center sample. The same residual KLT serves both the sum and the difference branch. Section~\ref{sec:odd} explains why. As $\rho\to1$,~\eqref{eq:odd-intro} degenerates to the parity splitting of the odd-length DCT-II into the DCT-VI and DST-VII---the fixed-transform identity observed in~\cite{ReznikICASSP}. The DST-VII is thus the $\rho\to1$ limit of the residual KLT~\cite{HanSaxenaRose,ReznikICASSP}. The limit also leaves an identity of independent interest: an explicit orthogonal conversion between the $\lfloor N/2\rfloor$-point DST-VII and the $(\lfloor N/2\rfloor{+}1)$-point DCT-VI (Corollary~\ref{cor:conversion}).

Together,~\eqref{eq:even-intro} and~\eqref{eq:odd-intro} recurse through every order $N\ge2$, terminating at the two-point butterfly $\W_2$ and $\W_1=(1)$. Combining them, the residual KLT recurses through half order as well (Corollary~\ref{cor:resfam}). Two secondary results follow. Unrolling the recursion at short lengths yields complete, exact KLT modules for $N=2,\dots,8$ (Section~\ref{sec:short}). The factorization constants prove to be algebraic functions of $\rho$, in radicals through $N=8$. An open-source C reference implementation---of the short-length modules and of the general-$N$ recursions---together with the verification procedures, is available online~\cite{AR1KLTrepo}.

A companion conference paper~\cite{ReznikSPIE} takes a complementary, non-recursive route to the same transform: it computes the exact AR(1) KLT as a fixed DCT-II followed by $\rho$-dependent correction stages. The factorizations of the present paper are of a different kind---the KLT is expressed through itself---and it is this self-similarity that carries the recursion and the complexity results above.

The paper is organized as follows. Section~\ref{sec:preliminaries} sets notation and collects the classical ingredients. Sections~\ref{sec:even} and~\ref{sec:odd} prove the even- and odd-order factorizations and their trigonometric endpoints. Section~\ref{sec:short} lists the short-length algorithms, Section~\ref{sec:complexity} discusses fast application and complexity, and Section~\ref{sec:conclusions} concludes.

\section{Preliminaries}
\label{sec:preliminaries}

\subsection{Notation}

Matrices are boldface capitals carrying their order as a subscript. $\I_N$ and $\J_N$ are the identity and counteridentity (order-reversal); empty positions in block matrices are zero. The symbol $\oplus$ denotes the direct (block-diagonal) sum. Throughout, $N$ is the transform length and $M$ the half order, $M=N/2$ for even $N$ and $M=\lfloor N/2\rfloor$ for odd $N$ (so $N=2M$ or $2M+1$). Vectors are boldface lowercase. $\e_k$ is the $k$-th unit coordinate vector, indexed from $0$. Both factorizations begin by folding a length-$N$ block about its center into half-size blocks of sums and differences of mirror-symmetric samples. Superscripts $(s)$ and $(a)$ label these symmetric (sum) and antisymmetric (difference) halves, and the last coordinate $\e_{M-1}$ of a half-size block, nearest the fold, is the \emph{fold coordinate}. A vector $\yb$ is symmetric if $\J_N\yb=\yb$ and antisymmetric if $\J_N\yb=-\yb$.

\subsection{Trigonometric transforms}
\label{sec:trig}

The following classical trigonometric transforms will arise as special cases (cf.~\cite{Britanak}). Rows $n$ and columns $k$ are indexed from $0$. The even-size family, for order $N$---the DCT-II, the DCT-IV, and the DST-I:
\begin{equation}
\begin{aligned}
\bigl[\Cb^{\mathrm{II}}_{N}\bigr]_{n,k} &= \sqrt{\tfrac{2}{N}}\,\beta_n\cos\tfrac{\pi n(2k+1)}{2N},\\
\bigl[\Cb^{\mathrm{IV}}_{N}\bigr]_{n,k} &= \sqrt{\tfrac{2}{N}}\,\cos\tfrac{\pi(2n+1)(2k+1)}{4N},\\
\bigl[\Sb^{\mathrm{I}}_{N}\bigr]_{n,k} &= \sqrt{\tfrac{2}{N+1}}\,\sin\tfrac{\pi(n+1)(k+1)}{N+1}.
\end{aligned}
\label{eq:trigeven}
\end{equation}
The odd family, on the frequency grid $\pi/(2M+1)$---the DST-VII, the DCT-VI, and the DCT-VIII:
\begin{equation}
\begin{aligned}
\bigl[\Sb^{\mathrm{VII}}_M\bigr]_{n,k}
&=\tfrac{2}{\sqrt{2M+1}}\sin\tfrac{\pi(2n+1)(k+1)}{2M+1},\\
\bigl[\Cb^{\mathrm{VI}}_{M+1}\bigr]_{n,k}
&=\tfrac{2}{\sqrt{2M+1}}\,\beta_n\gamma_k\cos\tfrac{\pi n(2k+1)}{2M+1},\\
\bigl[\Cb^{\mathrm{VIII}}_M\bigr]_{n,k}
&=\tfrac{2}{\sqrt{2M+1}}\cos\tfrac{\pi(2n+1)(2k+1)}{2(2M+1)},
\end{aligned}
\label{eq:trigodd}
\end{equation}
with $\beta_0=\gamma_M=1/\sqrt2$ and $\beta_n=\gamma_k=1$ otherwise.

\subsection{The AR(1) source, its covariance, and the KLT}
\label{sec:source}

Let $(x_k)$ be a stationary zero-mean, unit-variance AR(1) source, $x_k=\rho\,x_{k-1}+\sqrt{1-\rho^2}\,z_k$ with $(z_k)$ unit-variance white noise and $\rho\in[0,1)$. A block $\xb=(x_0,\dots,x_{N-1})^{\mathsf T}$ has covariance
\begin{equation}
\R_N(\rho) \;\triangleq\; E\!\left[\xb\xb^{\mathsf T}\right] \;=\; \bigl[\rho^{|i-j|}\bigr]_{i,j=0}^{N-1},
\label{eq:covariance}
\end{equation}
the Kac--Murdock--Szeg\H{o} matrix~\cite{KMS,GrenanderSzego}. It is symmetric, positive-definite, Toeplitz, and \emph{centrosymmetric}: $R_{N-1-i,\,N-1-j}=R_{i,j}$.

The KLT is defined by the eigenproblem of the covariance matrix $\R_N(\rho)$. Let $\wb^{(m)}$, $m=1,\dots,N$, be its unit eigenvectors,
\begin{equation}
\R_N(\rho)\,\wb^{(m)} = \lambda_m\,\wb^{(m)}, \qquad
\lambda_1 \ge \dots \ge \lambda_N > 0.
\label{eq:eigenproblem}
\end{equation}
Then $\W_N(\rho)$ is the orthogonal matrix whose $m$-th row is $\wb^{(m)\mathsf T}$. The coefficients $\yb=\W_N\xb$ are uncorrelated, with variances sorted in decreasing order. Among orthogonal transforms, this basis is optimal for quantization and encoding~\cite{KramerMathews,HuangSchultheiss}. It is this matrix, at the actual $\rho$ of the source, that we factorize. Each row is determined by~\eqref{eq:eigenproblem} only up to sign, and all statements below hold modulo this standard row-sign freedom (Remark~\ref{rem:signs}). At $\rho=0$ the spectrum is degenerate and the KLT non-unique---hence the standing assumption $\rho\in(0,1)$, with the DST-I limit attained as $\rho\to0^+$.

Centrosymmetry determines how many eigenvectors of the covariance are symmetric and how many are antisymmetric. By Cantoni and Butler~\cite{CantoniButler}, a symmetric centrosymmetric matrix of order $N$ admits an orthonormal eigenbasis that splits by parity: $\lceil N/2\rceil$ of its eigenvectors can be chosen symmetric, and the remaining $\lfloor N/2\rfloor$ antisymmetric. Moreover---and this is the case relevant here---if the matrix is also \emph{tridiagonal} and its eigenvalues are distinct, then the eigenvectors, ordered by eigenvalue, \emph{alternate} in parity~\cite{CantoniButler}. We use this alternation, anchored by Proposition~\ref{prop:phase} below, to interleave the sum and difference branches of the even- and odd-order factorizations derived in Sections~\ref{sec:even} and~\ref{sec:odd}.

\subsection{The generator}
\label{sec:generator}

The inverse of~\eqref{eq:covariance} is tridiagonal~\cite{JainFastKLT}; we normalize it as
\begin{equation}
\T_N(\rho) \triangleq (1-\rho^2)\,\R_N(\rho)^{-1}
=
\setlength{\arraycolsep}{3pt}
\begin{pmatrix}
1 & -\rho & & \\
-\rho & 1+\rho^2 & \ddots & \\
& \ddots & \ddots & -\rho \\
& & -\rho & 1
\end{pmatrix},
\label{eq:generator}
\end{equation}
and call $\T_N$ the \emph{generator} of the source. Eigenvectors of $\R_N$ are eigenvectors of $\T_N$, with reciprocally ordered eigenvalues
\begin{equation}
\mu=(1-\rho^2)/\lambda.
\label{eq:mulambda}
\end{equation} The display is for $N\ge2$; we set $\T_1\triangleq(1-\rho^2)$ so that the defining relation holds at every order (this matters only in the terminal case $M=1$ of the odd-order recursion). The generator is the natural domain for what follows. It is banded, centrosymmetric, and tridiagonal, so the Cantoni--Butler alternation applies to it. And, as we will see repeatedly, perturbations that are dense in covariance domain become perturbations of a single boundary entry in generator domain.

\subsection{Eigenstructure and the phase equation}
\label{sec:eigenstructure}

The exact eigenstructure of~\eqref{eq:covariance} is classical in two equivalent formulations. In continuous time, Ray and Driver~\cite{RayDriver} studied the process with exponential covariance on a symmetric interval $[-T,T]$. Its Karhunen--Lo\`eve eigenfunctions are sinusoids: even modes $\cos\omega t$ at frequencies satisfying $\tan(\omega T)=\beta/\omega$, and odd modes $\sin\omega t$ at frequencies satisfying $\tan(\omega T)=-\omega/\beta$. The discrete problem behaves the same way. Each eigenvector of $\R_N(\rho)$ is a sampled sinusoid at some frequency $\omega\in(0,\pi)$. The admissible frequencies are the roots of the transcendental equation
\begin{equation}
\tan(N\omega) \;=\; -\,\frac{(1-\rho^2)\sin\omega}{(1+\rho^2)\cos\omega-2\rho},
\label{eq:classical}
\end{equation}
and the eigenvalue attached to a frequency $\omega$ is $\lambda(\omega)=(1-\rho^2)/(1-2\rho\cos\omega+\rho^2)$~\cite{KMS,JainBook,ClarkeDCT,Britanak}.

We organize the same solution around a single monotone \emph{phase equation}, in which each end of the block contributes one explicit phase term. The short proof also fixes conventions used throughout.

\begin{proposition}[Phase equation]
\label{prop:phase}
Define the boundary phase
\begin{equation}
\theta(\omega) \triangleq \arg\bigl(1-\rho e^{-j\omega}\bigr) = \arctan\frac{\rho\sin\omega}{1-\rho\cos\omega}
\label{eq:phase-def}
\end{equation}
on $\omega\in(0,\pi)$. For each $m=1,\dots,N$ the equation
\begin{equation}
(N+1)\,\omega + 2\,\theta(\omega) \;=\; m\pi
\label{eq:phase}
\end{equation}
has a unique root $\omega_m\in(0,\pi)$, and $\omega_1<\dots<\omega_N$. The eigenpairs of $\R_N(\rho)$ are
\begin{equation}
\lambda_m = \frac{1-\rho^2}{1-2\rho\cos\omega_m+\rho^2},
\label{eq:eigenvalues}
\end{equation}
\begin{equation}
w^{(m)}_k = c_m \sin\bigl(\omega_m(k+1)+\theta(\omega_m)\bigr),
\label{eq:eigenvectors}
\end{equation}
for $k=0,\dots,N-1$, with $\lambda_1>\dots>\lambda_N$ and normalization constants $c_m$. Moreover, the eigenvector $\wb^{(m)}=(w^{(m)}_0,\dots,w^{(m)}_{N-1})^{\mathsf T}$ is symmetric for odd $m$ and antisymmetric for even $m$.
\end{proposition}

\begin{proof}
An interior row of $\T_N\wb=\mu\wb$ reads $-\rho w_{k-1}+(1+\rho^2)w_k-\rho w_{k+1}=\mu w_k$. Substituting $w_k=\sin(\omega k+\phi)$ and using $\sin(\omega(k{-}1)+\phi)+\sin(\omega(k{+}1)+\phi)=2\cos\omega\,\sin(\omega k+\phi)$ turns every interior row into an identity, provided $\mu=1-2\rho\cos\omega+\rho^2$---equivalently, by~\eqref{eq:mulambda}, provided the covariance eigenvalue is $\lambda(\omega)=(1-\rho^2)/\mu$, as in~\eqref{eq:eigenvalues}. The first and last rows amount to extending the recurrence with fictitious samples obeying the boundary conditions
\begin{equation}
w_{-1}=\rho\,w_0, \qquad w_N=\rho\,w_{N-1}.
\label{eq:robin}
\end{equation}
The left condition reads $\sin(\phi-\omega)=\rho\sin\phi$, i.e., $\operatorname{Im}[e^{j\phi}(e^{-j\omega}-\rho)]=0$; hence, modulo $\pi$ (a row sign), $\phi=\arg(e^{j\omega}-\rho)=\omega+\theta(\omega)$, so $w_k=\sin(\omega(k+1)+\theta)$; substituting into the right condition and writing $\alpha=\omega(N+1)+\theta$ yields $\tan\alpha=-\tan\theta$, i.e.\ \eqref{eq:phase}. Since $\theta'(\omega)\ge-\rho/(1+\rho)>-\tfrac12$, the left side of~\eqref{eq:phase} is strictly increasing with slope exceeding $N$, running from $0$ to $(N+1)\pi$ on $(0,\pi)$: each level $m\pi$, $m=1,\dots,N$, is crossed exactly once, and since $\lambda(\omega)$ strictly decreases, the spectrum is simple. Normalizing $\wb^{(m)}$ to unit length fixes $c_m$ in~\eqref{eq:eigenvectors}. Finally, by~\eqref{eq:phase}, $w^{(m)}_{N-1-k}=c_m\sin(m\pi-\omega_m(k+1)-\theta)=(-1)^{m+1}w^{(m)}_k$.
\end{proof}

Each eigenvector $\wb^{(m)}$ is a sampled sinusoid reflecting between the block ends, each end contributing the phase $\theta(\omega_m)$ to the quantized total $m\pi$ in~\eqref{eq:phase}; we refer to these sinusoidal eigenvectors as the \emph{modes} of the block. Splitting~\eqref{eq:phase} by the parity of $m$ makes the connection to Ray and Driver exact:

\begin{corollary}[Parity frequency equations]
\label{cor:parity}
The frequencies of the symmetric modes (odd $m$) are exactly the roots in $(0,\pi)$ of
\begin{equation}
\cos\tfrac{(N+1)\omega}{2}-\rho\cos\tfrac{(N-1)\omega}{2}=0;
\label{eq:parity-sym}
\end{equation}
those of the antisymmetric modes (even $m$) the roots of
\begin{equation}
\sin\tfrac{(N+1)\omega}{2}-\rho\sin\tfrac{(N-1)\omega}{2}=0.
\label{eq:parity-antisym}
\end{equation}
\end{corollary}

\begin{proof}
Equation~\eqref{eq:phase} reads $\tfrac{(N+1)\omega}{2}+\theta=\tfrac{m\pi}{2}$. Multiplying by $|1-\rho e^{-j\omega}|$ and using $e^{j\theta}|1-\rho e^{-j\omega}|=1-\rho e^{-j\omega}$, the real part of $e^{j\frac{(N+1)\omega}{2}}(1-\rho e^{-j\omega})$ vanishes for odd $m$, which is~\eqref{eq:parity-sym}, and the imaginary part for even $m$, which is~\eqref{eq:parity-antisym}. Expanding at $\tfrac{N\omega}{2}$ and $\tfrac\omega2$ gives the equivalent tangent forms $\tan\tfrac{N\omega}{2}\tan\tfrac{\omega}{2}=\tfrac{1-\rho}{1+\rho}$ and $\tan\tfrac{N\omega}{2}=-\tfrac{1+\rho}{1-\rho}\tan\tfrac{\omega}{2}$, the discrete counterparts of the Ray--Driver pair.
\end{proof}

For context, we note two classical facts. First, the product of the left-hand sides of~\eqref{eq:parity-sym} and~\eqref{eq:parity-antisym} equals $\tfrac12\bigl[\sin((N{+}1)\omega)-2\rho\sin(N\omega)+\rho^2\sin((N{-}1)\omega)\bigr]$, and setting it to zero rearranges to~\eqref{eq:classical}: the classical characteristic function is the \emph{product} of the two parity conditions. This is why the parity split---present already in Ray and Driver's own solution---is invisible in~\eqref{eq:classical}, and appears to have gone unexploited. Second, in the limits the roots form uniform grids: at $\rho=0$ the modes are the DST-I basis, and as $\rho\to1$ the DCT-II basis~\cite{ClarkeDCT,ClarkeDST}; see~\cite{Strang99} for the boundary-condition view of the trigonometric bases.

\subsection{Secular tools: rank-one and arrowhead updates}
\label{sec:secular}

The factorizations of Sections~\ref{sec:even} and~\ref{sec:odd} will reduce the KLT eigenproblem to small structured updates of problems already diagonalized. Three classical facts handle such updates~\cite{ShermanMorrison,Golub73,BunchNielsenSorensen,OLearyStewart}; see~\cite{Golub73} for a unified treatment.

\emph{(i) Sherman--Morrison}~\cite{ShermanMorrison}: for invertible $\mathbf{A}$ and $1+\vb^{\mathsf T}\mathbf{A}^{-1}\ub\neq0$,
\begin{equation}
\bigl(\mathbf{A}+\ub\vb^{\mathsf T}\bigr)^{-1}
= \mathbf{A}^{-1} - \frac{\mathbf{A}^{-1}\ub\,\vb^{\mathsf T}\mathbf{A}^{-1}}{1+\vb^{\mathsf T}\mathbf{A}^{-1}\ub}.
\label{eq:sm}
\end{equation}
Its role here is structural: it converts rank-one perturbations of a covariance into rank-one perturbations of its generator, and identifies the coordinate that supports them.

\emph{(ii) Rank-one update}~\cite{Golub73,BunchNielsenSorensen}: for $\G=\Db+\sigma\pb\pb^{\mathsf T}$ with $\Db=\operatorname{diag}(d_1<\dots<d_M)$, $\|\pb\|=1$, all $p_i\neq0$, $\sigma\neq0$, the eigenvalues $\nu_1<\dots<\nu_M$ of $\G$ are the roots of the secular function
\begin{equation}
f(\nu)=1+\sigma\sum_{i}\frac{p_i^2}{d_i-\nu},
\label{eq:secular}
\end{equation}
strictly interlacing the poles ($d_i<\nu_i<d_{i+1}$ for $\sigma>0$; $d_{i-1}<\nu_i<d_i$ for $\sigma<0$), and the unit eigenvector $\wb(i)$ of $\G$ associated with $\nu_i$ satisfies
\begin{equation}
\wb(i)\;\propto\;(\Db-\nu_i\I)^{-1}\pb .
\label{eq:secular-vectors}
\end{equation}
This is the pivotal object of the divide-and-conquer eigensolvers~\cite{Cuppen,DongarraSorensen}.

\emph{(iii) Arrowhead}~\cite{Golub73,OLearyStewart}: for
$\Hb=\bigl(\begin{smallmatrix}\Db&\bb\\ \bb^{\mathsf T}&\alpha\end{smallmatrix}\bigr)$
with $\Db=\operatorname{diag}(d_1<\dots<d_M)$ and all $b_i\neq0$, the eigenvalues $\nu_1<\dots<\nu_{M+1}$ are the roots of the strictly decreasing secular function $g(\nu)=\alpha-\nu-\sum_i b_i^2/(d_i-\nu)$, strictly interlacing the poles \emph{from outside},
\begin{equation}
\nu_1 < d_1 < \nu_2 < \dots < d_M < \nu_{M+1},
\label{eq:arrowhead-interlacing}
\end{equation}
and the unit eigenvector $\wb(i)$ of $\Hb$ associated with $\nu_i$ satisfies
\begin{equation}
\wb(i)\;\propto\;\begin{pmatrix}(\Db-\nu_i\I)^{-1}\bb\\ -1\end{pmatrix}.
\label{eq:arrowhead-vectors}
\end{equation}
In this paper the rank-one update diagonalizes the branches of the even fold and the difference branch of the odd fold; the arrowhead absorbs the center sample of the odd fold.

\section{The even-order factorization}
\label{sec:even}

Let $N$ be even, $M=N/2$. The derivation has three steps: fold the covariance about its center (Lemma~\ref{lem:fold-even}), pass to the generator, where the perturbation localizes to one entry (Corollary~\ref{cor:gen-even}), and diagonalize the resulting rank-one updates in the basis of the half-size KLT (Theorem~\ref{thm:even}). We begin with the fold. The parity decomposition $u_k=(x_k+x_{N-1-k})/\sqrt2$, $v_k=(x_k-x_{N-1-k})/\sqrt2$, $k=0,\dots,M-1$, is realized by the orthonormal butterfly
\begin{equation}
\B_N \;=\; \frac{1}{\sqrt2}
\begin{pmatrix}
\I_M & \J_M \\
\I_M & -\J_M
\end{pmatrix}.
\label{eq:butterfly-even}
\end{equation}
That $\B_N$ block-diagonalizes any symmetric centrosymmetric matrix is classical~\cite{CantoniButler}; the content of the next lemma is the \emph{form} the blocks take for the AR(1) covariance.

\begin{lemma}[Fold of the covariance, even order]
\label{lem:fold-even}
Let $\vb=\sqrt{\rho}\,\R_M(\rho)\,\e_{M-1}$, i.e.\ $v_k=\rho^{\,M-k-1/2}$. Then
\begin{equation}
\B_N\,\R_N(\rho)\,\B_N^{\mathsf T}
=\bigl(\R_M(\rho)+\vb\vb^{\mathsf T}\bigr)\oplus\bigl(\R_M(\rho)-\vb\vb^{\mathsf T}\bigr).
\label{eq:fold-even}
\end{equation}
\end{lemma}

\begin{proof}
By centrosymmetry the fold gives blocks $R_{ij}\pm R_{i,N-1-j}$, $0\le i,j\le M-1$~\cite{CantoniButler}. Since $i+j\le N-2$, the wrap term is $R_{i,N-1-j}=\rho^{N-1-i-j}=\rho^{\,M-i-1/2}\rho^{\,M-j-1/2}=v_iv_j$; and $[\sqrt\rho\,\R_M\e_{M-1}]_k=\sqrt\rho\,\rho^{\,M-1-k}=v_k$.
\end{proof}

Probabilistically, the rank-one term is the wrap-around correlation between a sample and the mirrored block. Structurally, what matters is that the perturbation vector is proportional to the \emph{last column of the half-size covariance itself}. Sherman--Morrison then localizes it exactly:

\begin{corollary}[Generator-domain form]
\label{cor:gen-even}
With $\T^{(s)}_M\triangleq(1-\rho^2)(\R_M+\vb\vb^{\mathsf T})^{-1}$, $\T^{(a)}_M\triangleq(1-\rho^2)(\R_M-\vb\vb^{\mathsf T})^{-1}$, so that $\B_N\T_N\B_N^{\mathsf T}=\T^{(s)}_M\oplus\T^{(a)}_M$,
\begin{equation}
\begin{aligned}
\T^{(s)}_M &= \T_M(\rho)-\sigma_s\,\e_{M-1}\e_{M-1}^{\mathsf T},\\
\T^{(a)}_M &= \T_M(\rho)+\sigma_a\,\e_{M-1}\e_{M-1}^{\mathsf T},
\end{aligned}
\label{eq:gen-even}
\end{equation}
with the perturbation strengths
\begin{equation}
\sigma_s=\rho(1-\rho),\qquad \sigma_a=\rho(1+\rho).
\label{eq:strengths}
\end{equation}
\end{corollary}

\begin{proof}
Since $\vb=\sqrt\rho\,\R_M\e_{M-1}$: $\R_M^{-1}\vb=\sqrt\rho\,\e_{M-1}$ and $\vb^{\mathsf T}\R_M^{-1}\vb=\rho$. By~\eqref{eq:sm} with $\ub=\pm\vb$, $(\R_M\pm\vb\vb^{\mathsf T})^{-1}=\R_M^{-1}\mp\tfrac{\rho}{1\pm\rho}\e_{M-1}\e_{M-1}^{\mathsf T}$; multiply by $1-\rho^2$.
\end{proof}

This is the crux. Although the covariance perturbation is dense, the generator perturbation is a single scalar at the fold coordinate. Each branch is thus a genuine half-size AR(1) problem with one modified boundary condition, and its eigenproblem is a rank-one secular update.

Let $\W_M\T_M\W_M^{\mathsf T}=\M_M=\operatorname{diag}(\mu_1<\dots<\mu_M)$ be the half-size diagonalization and define the \emph{fold vector}
\begin{equation}
\pb=\W_M\e_{M-1},
\label{eq:foldvector}
\end{equation}
the vector of last components of the half-size eigenvectors, so that $\|\pb\|=1$. The hypotheses of~\eqref{eq:secular}--\eqref{eq:secular-vectors} hold here. Since the off-diagonal entries of the tridiagonal $\T_M$ equal $-\rho\neq0$, an eigenvector with vanishing last component would have every component vanish, row by row of the eigen-equation. Hence all $p_m\neq0$. And since two independent eigenvectors sharing an eigenvalue could be combined into one with vanishing last component, the eigenvalues $\mu_m$ are simple, which justifies the strict ordering. Conjugating~\eqref{eq:gen-even} by $\W_M$ leaves
\begin{equation}
\G^{(s)}_M = \M_M-\sigma_s\,\pb\pb^{\mathsf T},
\qquad
\G^{(a)}_M = \M_M+\sigma_a\,\pb\pb^{\mathsf T},
\label{eq:G}
\end{equation}
instances of~\eqref{eq:secular}. Denote by $\nu^{(s)}_1<\dots<\nu^{(s)}_M$ and $\nu^{(a)}_1<\dots<\nu^{(a)}_M$ their eigenvalues---the roots of the secular equation~\eqref{eq:secular} with $(\Db,\sigma)=(\M_M,-\sigma_s)$ and $(\M_M,\sigma_a)$, respectively. The matrices $\Qh^{(s)}_M$ and $\Qh^{(a)}_M$ are then the orthogonal matrices whose $i$-th rows are the unit eigenvectors~\eqref{eq:secular-vectors} associated with $\nu^{(s)}_i$ and $\nu^{(a)}_i$; entrywise, for $i,m=1,\dots,M$,
\begin{equation}
\bigl[\Qh^{(s)}_M\bigr]_{i,m}=c^{(s)}_i\,\frac{p_m}{\mu_m-\nu^{(s)}_i},
\quad
\bigl[\Qh^{(a)}_M\bigr]_{i,m}=c^{(a)}_i\,\frac{p_m}{\mu_m-\nu^{(a)}_i},
\label{eq:Qrows}
\end{equation}
where $c^{(s)}_i,c^{(a)}_i>0$ normalize each row to unit length. By construction, $\Qh^{(s)}_M\G^{(s)}_M\Qh^{(s)\mathsf T}_M=\operatorname{diag}(\nu^{(s)}_1,\dots,\nu^{(s)}_M)$, and likewise for the antisymmetric branch. All these quantities depend only on $\rho$ and are computed offline.

\begin{theorem}[Self-similar factorization, even order]
\label{thm:even}
For even $N=2M$ and $\rho\in(0,1)$,
\begin{equation}
\W_N(\rho) = \PP_N\Bigl[\bigl(\Qh^{(s)}_M\W_M(\rho)\bigr)\oplus\bigl(\Qh^{(a)}_M\W_M(\rho)\bigr)\Bigr]\B_N ,
\label{eq:even-thm}
\end{equation}
where $\PP_N$ interleaves ($[\PP_N]_{2i,\,i}=[\PP_N]_{2i+1,\,M+i}=1$): the even-indexed KLT outputs are the sum branch in order, the odd-indexed outputs the difference branch in order.
\end{theorem}

\begin{proof}
By Lemma~\ref{lem:fold-even}, Corollary~\ref{cor:gen-even}, and~\eqref{eq:G}, the matrix on the right without $\PP_N$ conjugates $\T_N$ to a diagonal matrix. Its rows are thus a complete orthonormal eigenvector system of $\T_N$, and hence of $\R_N$: the first $M$ rows symmetric, the last $M$ antisymmetric, each half ordered by decreasing $\R_N$-eigenvalue. Since $\T_N$ is tridiagonal centrosymmetric with distinct eigenvalues, the parities alternate along the ordered spectrum~\cite{CantoniButler}. By Proposition~\ref{prop:phase}, the alternation starts at the top of the $\R_N$-spectrum with a symmetric mode ($m=1$ odd). The permutation realizing this interleaving of the two internally ordered halves is $\PP_N$.
\end{proof}

\begin{remark}[Sign conventions]
\label{rem:signs}
Flipping the sign of row $m$ of $\W_M$ flips $p_m$ and a column-sign pattern inside $\Qh^{(s)},\Qh^{(a)}$, leaving the products in~\eqref{eq:even-thm} invariant. Any fixed convention may therefore be used, and~\eqref{eq:even-thm} determines $\W_N$ up to the usual row-sign freedom. The same holds for Theorem~\ref{thm:odd}.
\end{remark}

We now identify the $\rho\to1$ endpoint of the factorization, in terms of the transforms~\eqref{eq:trigeven}.

\begin{corollary}[$\rho\to1$: Chen's split of the DCT-II]
\label{cor:even-limit}
As $\rho\to1$ (up to row signs), $\Qh^{(s)}_M\W_M\to\Cb^{\mathrm{II}}_M$ and $\Qh^{(a)}_M\W_M\to\Cb^{\mathrm{IV}}_M$ in the fold coordinates, and~\eqref{eq:even-thm} degenerates to the classical splitting of the DCT-II into half-size DCT-II and DCT-IV~\cite{ChenSmithFralick,Wang84,Britanak},
\[
\Cb^{\mathrm{II}}_N \;=\; \PP_N\bigl(\Cb^{\mathrm{II}}_{M}\oplus\Cb^{\mathrm{IV}}_{M}\bigr)\B_N .
\]
Recursive application collapses the recursion of Theorem~\ref{thm:even} to the decomposition underlying Chen's fast DCT~\cite{ChenSmithFralick}.
\end{corollary}

\begin{proof}
At $\rho=1$ both corner entries of $\T^{(s)}_M$ equal $1$ (free ends): its modes are $\cos\bigl(\pi n(k+\tfrac12)/M\bigr)$, the DCT-II. The corners of $\T^{(a)}_M$ are $1$ and $1+\sigma_a=3$, i.e.\ the fictitious sample obeys $y_M=-y_{M-1}$: its modes are $\cos\bigl(\pi(2n+1)(k+\tfrac12)/(2M)\bigr)$, the DCT-IV. Both ascending eigenvalue orders coincide with the index orders, and continuity of the simple eigensystems in $\rho$ completes the proof.
\end{proof}

The two strengths~\eqref{eq:strengths} are asymmetric: $\sigma_s\to0$ but $\sigma_a\to2$ as $\rho\to1$. This is the exact-KLT form of the asymmetry of the classical split, which pairs the DCT-II with the DCT-IV rather than with itself. Folding is asymptotically free on the sum branch at high correlation. On the difference branch, it always costs a bounded rotation.

\begin{remark}[Genuine recursion]
\label{rem:recursion}
The half-size transform in~\eqref{eq:even-thm} is the exact KLT of the same source at order $M$: if $M$ is even, Theorem~\ref{thm:even} applies again; if odd, Theorem~\ref{thm:odd} below. The recursion therefore continues at every order, terminating at $\W_2$---the $\rho$-independent $45^\circ$ butterfly---and at $\W_1=(1)$. For $N$ a power of two, it unrolls into $\log_2N$ butterfly levels and rank-one-generated rotations.
\end{remark}

\section{The odd-order factorization}
\label{sec:odd}

At the fixed-transform level, the parity split of the \emph{odd}-length DCT-II behaves differently from the even split. It produces not two copies of a common core but two different transforms of the odd trigonometric family: a DCT-VI on the symmetric half, which includes the center sample, and a DST-VII on the antisymmetric half~\cite{ReznikICASSP}. The DST-VII, in turn, has an independent life in coding theory. It arises as the $\rho\to1$ limit of the KLT of an AR(1) block \emph{with a known boundary sample}, which motivated its adoption for residual coding in modern video standards~\cite{HanSaxenaRose,ReznikICASSP,ZhaoVVC}. This section exhibits the exact transform behind that limit, at every $\rho$. Both branches of the odd fold are driven by a single operator: the exact KLT of the block conditioned on its boundary sample. The center sample enters through a rank-one border---the arrowhead update of Section~\ref{sec:secular}.

\subsection{The fold, its conditional reading, and the residual KLT}

Let $N=2M+1$. The parity decomposition now comprises $M$ mirrored sums, the center sample $c=x_M$, and $M$ mirrored differences, realized by the orthonormal butterfly
\begin{equation}
\B_N \;=\; \frac{1}{\sqrt2}
\begin{pmatrix}
\I_M & & \J_M \\
& \sqrt2 & \\
\I_M & & -\J_M
\end{pmatrix},
\label{eq:butterfly-odd}
\end{equation}
with output ordering $(\ub,c,\vb)$; the symmetric branch, of size $M+1$, is $(\ub,c)$.

\begin{lemma}[Fold of the covariance, odd order]
\label{lem:fold-odd}
Let $\wb=\rho\,\R_M(\rho)\,\e_{M-1}$, i.e.\ $w_k=\rho^{\,M-k}$. Then
\begin{equation}
\B_N\R_N\B_N^{\mathsf T}
=
\setlength{\arraycolsep}{3pt}
\begin{pmatrix}
\R_M+\wb\wb^{\mathsf T} & \sqrt2\,\wb \\
\sqrt2\,\wb^{\mathsf T} & 1
\end{pmatrix}
\oplus
\bigl(\R_M-\wb\wb^{\mathsf T}\bigr),
\label{eq:fold-odd}
\end{equation}
with $\R_M=\R_M(\rho)$ and $\R_N=\R_N(\rho)$.
\end{lemma}

\begin{proof}
As in Lemma~\ref{lem:fold-even}, the wrap term is $R_{i,N-1-j}=\rho^{2M-i-j}=w_iw_j$; the center couplings are $E[u_ic]=\tfrac{1}{\sqrt2}\cdot2\rho^{M-i}=\sqrt2\,w_i$, $E[v_ic]=0$, $E[c^2]=1$.
\end{proof}

\begin{remark}[Conditional reading: one process drives both branches]
\label{rem:conditional}
The vector $\wb$ is exactly the correlation of the block with its center sample: $\operatorname{Cov}(x_i,x_M)=\rho^{M-i}=w_i$ and $\operatorname{Var}(x_M)=1$. Hence $\R_M-\wb\wb^{\mathsf T}=\operatorname{Cov}\bigl((x_0,\dots,x_{M-1})\mid x_M\bigr)$: the covariance of an AR(1) block \emph{conditioned on---equivalently, after optimal one-sided prediction from---its boundary sample}. This one covariance governs both branches. It governs the differences outright, since $\vb$ is uncorrelated with $c$ by~\eqref{eq:fold-odd}. By the Gaussian conditioning formula, $\operatorname{Cov}(\ub\mid c)=(\R_M+\wb\wb^{\mathsf T})-2\wb\wb^{\mathsf T}=\R_M-\wb\wb^{\mathsf T}$ as well: \emph{conditioned on the center sample, the folded sums and differences are identically distributed}. This is the nontrivial counterpart of the even fold: there the two branches share the base covariance $\R_M$ by construction, so the shared inner operator $\W_M$ is expected. Here the unconditional branch covariances differ, and coincide only after conditioning.
\end{remark}

\begin{corollary}[Generator-domain form]
\label{cor:gen-odd}
With $\T^{(\delta)}_M\triangleq(1-\rho^2)(\R_M-\wb\wb^{\mathsf T})^{-1}$, the generator of the residual process,
\begin{equation}
\T^{(\delta)}_M = \T_M(\rho)+\rho^2\,\e_{M-1}\e_{M-1}^{\mathsf T},
\label{eq:A}
\end{equation}
and the fold block-diagonalizes the generator as
\begin{equation}
\begin{gathered}
\B_N\,\T_N\,\B_N^{\mathsf T}
=\Sb_{M+1}\oplus\T^{(\delta)}_M,\\
\Sb_{M+1}=
\setlength{\arraycolsep}{2pt}
\begin{pmatrix}
\T^{(\delta)}_M & -\sqrt2\rho\,\e_{M-1} \\
-\sqrt2\rho\,\e_{M-1}^{\mathsf T} & 1+\rho^2
\end{pmatrix}\!.
\end{gathered}
\label{eq:S}
\end{equation}
\end{corollary}

\begin{proof}
For $\T^{(\delta)}_M$: $\R_M^{-1}\wb=\rho\,\e_{M-1}$, $\wb^{\mathsf T}\R_M^{-1}\wb=\rho^2$, and~\eqref{eq:sm} with $\ub=-\wb$ gives $(\R_M-\wb\wb^{\mathsf T})^{-1}=\R_M^{-1}+\tfrac{\rho^2}{1-\rho^2}\e_{M-1}\e_{M-1}^{\mathsf T}$. For the block form, fold~\eqref{eq:generator} directly: the band kills all wrap terms ($|i-(N-1-j)|\ge2$ for $i,j\le M-1$), so both parity blocks equal the leading principal $M\times M$ submatrix of $\T_N$, which is $\T_M+\rho^2\e_{M-1}\e_{M-1}^{\mathsf T}=\T^{(\delta)}_M$ (interior value $1+\rho^2$ at the fold corner). The only center couplings are $\tfrac{1}{\sqrt2}(-\rho\mp\rho)$, i.e.\ $-\sqrt2\rho$ for the sums and $0$ for the differences, with $T_{M,M}=1+\rho^2$.
\end{proof}

Again a dense covariance perturbation collapses to one boundary scalar. But the modified boundary condition is now different in kind: replacing the corner $1$ of $\T_M$ by $1+\rho^2$ replaces the condition of~\eqref{eq:robin} by $w_M=0$---the zero that every antisymmetric mode of the full problem must place at the center sample. The strength is $\rho^2$, versus $\rho(1\mp\rho)$ for the even fold. The sum-branch generator contains $\T^{(\delta)}_M$ as its leading principal block---the precision-matrix face of Remark~\ref{rem:conditional}. The fold thus hands us one rank-one update and one rank-one \emph{border}: an arrowhead.

Let $\W_M$, $\M_M$, $\pb$ be as in Section~\ref{sec:even}. Conjugating $\T^{(\delta)}_M$ by $\W_M$ gives $\G^{(\delta)}_M=\M_M+\rho^2\pb\pb^{\mathsf T}$, an instance of~\eqref{eq:secular} with eigenvalues $\delta_1<\dots<\delta_M$, collected in $\boldsymbol{\Delta}^{(\delta)}_M=\operatorname{diag}(\delta_1,\dots,\delta_M)$. Let $\Qh^{(\delta)}_M$ be its orthogonal diagonalizer, with rows of the same Cauchy form as~\eqref{eq:Qrows}:
\begin{equation}
\bigl[\Qh^{(\delta)}_M\bigr]_{i,m}=c^{(\delta)}_i\,\frac{p_m}{\mu_m-\delta_i},
\qquad i,m=1,\dots,M,
\label{eq:Qdrows}
\end{equation}
with $c^{(\delta)}_i>0$ normalizing each row to unit length. Define the \emph{residual KLT}
\begin{equation}
\V_M(\rho) \;\triangleq\; \Qh^{(\delta)}_M\,\W_M(\rho),
\qquad
\V_M\,\T^{(\delta)}_M\,\V_M^{\mathsf T}=\boldsymbol{\Delta}^{(\delta)}_M .
\label{eq:V}
\end{equation}
Its rows are the unit eigenvectors of the conditional covariance, ordered by decreasing variance $(1-\rho^2)/\delta_m$. That is, $\V_M$ is, literally, the KLT of the one-sided prediction residual---reached from the unconditional half-size KLT by one secular rotation. Since $\T^{(\delta)}_M$ is also tridiagonal with nonzero off-diagonals for $\rho>0$, the argument of Section~\ref{sec:even} applies again: its spectrum is simple, and the entries of $\qb\triangleq\V_M\e_{M-1}$ are all nonzero.

\subsection{The factorization theorem and its endpoints}

All the pieces are now in place. Conjugating $\Sb_{M+1}$ by $\V_M\oplus1$ leaves exactly an arrowhead,
\begin{equation}
(\V_M{\oplus}1)\Sb_{M+1}(\V_M{\oplus}1)^{\mathsf T}\!
=\Hb_{M+1}\triangleq
\setlength{\arraycolsep}{1pt}
\begin{pmatrix}
\boldsymbol{\Delta}^{(\delta)}_M & -\sqrt2\rho\qb \\
-\sqrt2\rho\qb^{\mathsf T} & 1{+}\rho^2
\end{pmatrix}\!,
\label{eq:H}
\end{equation}
an instance of~\eqref{eq:arrowhead-interlacing}--\eqref{eq:arrowhead-vectors}. Let $\A_{M+1}$ be its orthogonal diagonalizer: the $i$-th row of $\A_{M+1}$ is the unit eigenvector~\eqref{eq:arrowhead-vectors} of $\Hb_{M+1}$ associated with its $i$-th smallest eigenvalue $\nu_i$.

\begin{theorem}[Self-similar factorization, odd order]
\label{thm:odd}
For odd $N=2M+1$ and $\rho\in(0,1)$,
\begin{equation}
\W_N(\rho) = \PP_N\bigl[\A_{M+1}\bigl(\V_M(\rho)\oplus1\bigr)\,\oplus\,\V_M(\rho)\bigr]\B_N ,
\label{eq:odd-thm}
\end{equation}
with $\B_N$ the odd butterfly~\eqref{eq:butterfly-odd} and $\PP_N$ the interleaving permutation ($[\PP_N]_{2i,\,i}=1$, $i\le M$; $[\PP_N]_{2i+1,\,M+1+i}=1$, $i\le M-1$): even-indexed KLT outputs are the sum branch in order, odd-indexed the difference branch in order.
\end{theorem}

\begin{proof}
By Corollary~\ref{cor:gen-odd} and~\eqref{eq:V}--\eqref{eq:H}, the block matrix (without $\PP_N$) times $\B_N$ conjugates $\T_N$ to $\operatorname{diag}(\nu_1,\dots,\nu_{M+1})\oplus\operatorname{diag}(\delta_1,\dots,\delta_M)$. Its rows are therefore a complete orthonormal eigenvector system: the sum rows symmetric, the difference rows antisymmetric, each branch internally ordered. The global order is the parity alternation of Section~\ref{sec:source}. Here it also follows unconditionally from the structure: all entries of $-\sqrt2\rho\,\qb$ are nonzero and the $\delta_m$ are distinct, so the arrowhead interlacing~\eqref{eq:arrowhead-interlacing} gives $\nu_1<\delta_1<\nu_2<\dots<\delta_M<\nu_{M+1}$. The spectrum thus alternates, beginning and ending with the sum branch---exactly the interleaving $\PP_N$, and consistent with Proposition~\ref{prop:phase} (the $M{+}1$ symmetric modes are those of odd $m$).
\end{proof}

Since $\V_M=\Qh^{(\delta)}_M\W_M$, the recursion continues through $\W_M$ (Remark~\ref{rem:recursion}). Composing the same identities closes it within the residual family as well.

\begin{corollary}[Recursion within the residual family]
\label{cor:resfam}
For $\rho\in(0,1)$ and even $N=2M$,
\begin{equation}
\V_N=\Qh^{(\delta)}_N\PP_N\bigl[\bigl(\Qh^{(s)}_M\Qh^{(\delta)\mathsf T}_M\V_M\bigr)\oplus\bigl(\Qh^{(a)}_M\Qh^{(\delta)\mathsf T}_M\V_M\bigr)\bigr]\B_N,
\label{eq:resfam}
\end{equation}
while for odd $N=2M+1$,
\begin{equation}
\V_N=\Qh^{(\delta)}_N\PP_N\bigl[\A_{M+1}(\V_M{\oplus}1)\oplus\V_M\bigr]\B_N,
\label{eq:resfam-odd}
\end{equation}
with every stage of~\eqref{eq:resfam-odd} a single secular or arrowhead stage.
\end{corollary}

\begin{proof}
For~\eqref{eq:resfam}, substitute~\eqref{eq:V} at orders $N$ and $M$ into~\eqref{eq:even-thm}; for~\eqref{eq:resfam-odd}, left-multiply~\eqref{eq:odd-thm} by $\Qh^{(\delta)}_N$.
\end{proof}

We now identify the $\rho\to1$ endpoint of the factorization, in terms of the transforms~\eqref{eq:trigodd}.

\begin{corollary}[$\rho\to1$: the odd-length DCT-II split]
\label{cor:odd-limit}
As $\rho\to1$ (up to row signs): $\V_M(\rho)\to\Db_\pm\Sb^{\mathrm{VII}}_M\J_M$ with $\Db_\pm=\operatorname{diag}((-1)^n)$, and $\A_{M+1}(\V_M\oplus1)\to\Cb^{\mathrm{VI}}_{M+1}$ in the fold coordinates $(\ub,c)$, the weight $\gamma_M=1/\sqrt2$ arising automatically from the center normalization of~\eqref{eq:butterfly-odd}. Hence~\eqref{eq:odd-thm} degenerates to the parity splitting of the odd-length DCT-II into the DCT-VI and DST-VII~\cite{ReznikICASSP},
\[
\Cb^{\mathrm{II}}_N = \PP_N\bigl(\Cb^{\mathrm{VI}}_{M+1}\oplus\Db_\pm\Sb^{\mathrm{VII}}_M\J_M\bigr)\B_N .
\] In particular, the exact residual KLT $\V_M(\rho)$ is, at every $\rho\in(0,1)$, the transform whose $\rho\to1$ limit is the DST-VII.
\end{corollary}

\begin{proof}
At $\rho=1$ the residual generator $\T^{(\delta)}_M$ becomes the second-difference matrix with a free outer end and the constraint $y_M=0$ at the fold. Its modes are $y_k=\cos\bigl(\theta_n(k+\tfrac12)\bigr)$ with $\theta_n=(2n+1)\pi/(2M+1)$, eigenvalues $2-2\cos\theta_n$ ascending. Reversing coordinates and using $\theta_n(M+\tfrac12)=(n+\tfrac12)\pi$, $\cos\bigl(\theta_n(M-1-k+\tfrac12)\bigr)=(-1)^n\sin\bigl(\theta_n(k+1)\bigr)$: row $n$ of $\Sb^{\mathrm{VII}}_M$ up to the sign $(-1)^n$. For the sum branch, the symmetric modes of $\T_N(1)$ (the DCT-II generator) are $y_k=\cos\bigl(2\pi m(k+\tfrac12)/(2M+1)\bigr)$, $m=0,\dots,M$, ascending. The branch vector of a unit mode is $(\sqrt2y_0,\dots,\sqrt2y_{M-1},y_M)$. Writing out the constants---using $\cos\bigl(\pi m(2M+1)/(2M+1)\bigr)=(-1)^m$ at the center---gives exactly the rows of $\Cb^{\mathrm{VI}}_{M+1}$, including the half-weights $\beta_0,\gamma_M$. Continuity of the simple eigensystems in $\rho$ completes the proof.
\end{proof}

Taking the $\rho\to1$ limit \emph{inside} the factorization~\eqref{eq:odd-thm} yields one more consequence. The parity split of the odd-length DCT-II presents $\Cb^{\mathrm{VI}}_{M+1}$ and $\Sb^{\mathrm{VII}}_M$ as independent diagonal blocks and, by itself, implies no relation between them. In fact one determines the other: the limiting arrowhead factor is an explicit orthogonal conversion between the two transforms, in closed form.

\begin{corollary}[Conversion between the DST-VII and the DCT-VI]
\label{cor:conversion}
For every $M\ge1$, with $\omega_m=\tfrac{2\pi m}{2M+1}$ and $\theta_n=\tfrac{(2n+1)\pi}{2M+1}$,
\begin{equation}
\Cb^{\mathrm{VI}}_{M+1} = \bar{\A}_{M+1}\bigl(\Db_\pm\,\Sb^{\mathrm{VII}}_M\,\J_M\oplus1\bigr),
\label{eq:conversion}
\end{equation}
where $\bar{\A}_{M+1}$ is the orthogonal matrix with entries ($m=0,\dots,M$; $n=0,\dots,M-1$)
\begin{equation}
\begin{aligned}
\bigl[\bar{\A}_{M+1}\bigr]_{m,n} &= (-1)^{m+n}\,\frac{2\beta_m}{2M+1}\cdot\frac{\sin\theta_n}{\cos\omega_m-\cos\theta_n},\\
\bigl[\bar{\A}_{M+1}\bigr]_{m,M} &= (-1)^{m}\,\frac{\sqrt2\,\beta_m}{\sqrt{2M+1}},
\end{aligned}
\label{eq:conversion-entries}
\end{equation}
with $\beta_m$ as in~\eqref{eq:trigodd}. Equivalently, since $\Db_\pm\Sb^{\mathrm{VII}}_M\J_M=\Cb^{\mathrm{VIII}}_M$,
\begin{equation}
\Cb^{\mathrm{VI}}_{M+1}=\bar{\A}_{M+1}\bigl(\Cb^{\mathrm{VIII}}_M\oplus1\bigr).
\label{eq:conversion8}
\end{equation}
\end{corollary}

\begin{proof}
Relation~\eqref{eq:conversion} with \emph{some} orthogonal $\bar{\A}_{M+1}$ is the $\rho\to1$ limit of Theorem~\ref{thm:odd}. By Corollary~\ref{cor:odd-limit}, the sum branch $\A_{M+1}(\V_M\oplus1)$ tends to $\Cb^{\mathrm{VI}}_{M+1}$ while $\V_M\to\Db_\pm\Sb^{\mathrm{VII}}_M\J_M$; hence $\bar{\A}_{M+1}=\lim_{\rho\to1}\A_{M+1}$, and it remains to compute $\bar{\A}_{M+1}=\Cb^{\mathrm{VI}}_{M+1}(\Db_\pm\Sb^{\mathrm{VII}}_M\J_M\oplus1)^{\mathsf T}$ entrywise. Using $\theta_n(M+\tfrac12)=(n+\tfrac12)\pi$, the reversed rows simplify to $[\Db_\pm\Sb^{\mathrm{VII}}_M\J_M]_{n,k}=\tfrac{2}{\sqrt{2M+1}}\cos\bigl(\theta_n(k+\tfrac12)\bigr)$, which is~\eqref{eq:conversion8}. The interior entries are then
\[
\bigl[\bar{\A}_{M+1}\bigr]_{m,n}
=\frac{4\beta_m}{2M+1}\sum_{k=0}^{M-1}\cos\bigl(\omega_m(k+\tfrac12)\bigr)\cos\bigl(\theta_n(k+\tfrac12)\bigr).
\]
Expand the product into $\tfrac12\cos\bigl((\omega_m\mp\theta_n)(k+\tfrac12)\bigr)$ and apply the closed sum $\sum_{k=0}^{M-1}\cos\bigl(\alpha(k+\tfrac12)\bigr)=\tfrac{\sin(\alpha M)}{2\sin(\alpha/2)}$. Reducing $\sin(\alpha M)$ with $M=\tfrac{(2M+1)-1}{2}$ (the numerators of $\omega_m\mp\theta_n$ are odd multiples of $\tfrac{\pi}{2M+1}$) turns both terms into cotangents, whose difference collapses by product-to-sum to the stated Cauchy form. The center column is $[\Cb^{\mathrm{VI}}_{M+1}]_{m,M}=(-1)^m\sqrt2\beta_m/\sqrt{2M+1}$ directly from~\eqref{eq:trigodd}. Orthogonality of $\bar{\A}_{M+1}$ is inherited from~\eqref{eq:conversion}.
\end{proof}

The limit $\rho\to1$ can likewise be taken inside~\eqref{eq:resfam} and~\eqref{eq:resfam-odd}: the residual KLTs tend to the DST-VII, and the recursions of Corollary~\ref{cor:resfam} become self-similar factorizations of the DST-VII itself, with fixed Cauchy-structured stages.

\section{Exact short-length KLT factorizations}
\label{sec:short}

Unrolling Theorems~\ref{thm:even} and~\ref{thm:odd} at short lengths yields complete, exact KLT algorithms. We list them here for $N=2,\dots,8$, in the style of the classical short-length DCT modules (cf.~\cite{Heideman}): as chains of operations per stage, without proofs. Throughout, $G(\varphi)$ denotes the plane rotation
\[
G(\varphi):\ (a,b)\ \mapsto\ (a\cos\varphi+b\sin\varphi,\ -a\sin\varphi+b\cos\varphi).
\]
All constants are precomputed functions of $\rho\in(0,1)$. Outputs $y_0,\dots,y_{N-1}$ appear in KLT order (decreasing variance), with rows fixed up to sign (Remark~\ref{rem:signs}). Dense final stages are written as orthogonal matrices applied to intermediate vectors. Their rows are given in the Cauchy forms~\eqref{eq:secular-vectors} and~\eqref{eq:arrowhead-vectors}, normalized to unit length and ordered by ascending root. Every constant below is an elementary or radical function of $\rho$: the defining polynomials have degree $\le4$ for all $N\le8$, so their roots are expressible in radicals. The first length requiring a quintic equation---hence, in general, numerical secular roots---is $N=9$. As $\rho\to1$ the modules reduce to the corresponding fast DCT-II modules (for $N=4$, to Chen's~\cite{ChenSmithFralick}); as $\rho\to0$, to DST-I modules.

\medskip
\noindent\textbf{N = 2.}\quad
$y_0=(x_0+x_1)/\sqrt2$, \ $y_1=(x_0-x_1)/\sqrt2$ \ (independent of $\rho$).

\medskip
\noindent\textbf{N = 3.}
\begin{align*}
&\text{Stage 1:}\ \ u=\tfrac{x_0+x_2}{\sqrt2}, \quad v=\tfrac{x_0-x_2}{\sqrt2};\\
&\text{Stage 2:}\ \ (y_0,y_2)=G(\varphi)\,(u,\,x_1), \quad y_1=v;\\
&\qquad\qquad \tan2\varphi=2\sqrt2/\rho.
\end{align*}
The single $\rho$-dependent angle $\varphi=\tfrac12\arctan(2\sqrt2/\rho)$ decreases from $45^\circ$ ($\rho\to0$; 3-point DST-I) to $\tfrac12\arctan2\sqrt2\approx35.26^\circ$ ($\rho\to1$; 3-point DCT-II).

\medskip
\noindent\textbf{N = 4.}
\begin{align*}
&\text{Stage 1:}\ \ u_0=\tfrac{x_0+x_3}{\sqrt2},\ \ u_1=\tfrac{x_1+x_2}{\sqrt2},\\
&\qquad\qquad v_0=\tfrac{x_0-x_3}{\sqrt2},\ \ v_1=\tfrac{x_1-x_2}{\sqrt2};\\
&\text{Stage 2:}\ \ s_0=\tfrac{u_0+u_1}{\sqrt2},\ \ s_1=\tfrac{u_0-u_1}{\sqrt2},\\
&\qquad\qquad d_0=\tfrac{v_0+v_1}{\sqrt2},\ \ d_1=\tfrac{v_0-v_1}{\sqrt2};\\
&\text{Stage 3:}\ \ (y_0,y_2)=G(\varphi_s)\,(s_0,s_1),\\
&\qquad\qquad (y_1,y_3)=G(\varphi_a)\,(d_0,d_1);\\
&\qquad\qquad \tan2\varphi_s=-\tfrac{1-\rho}{2}, \quad \tan2\varphi_a=\tfrac{1+\rho}{2}.
\end{align*}
As $\rho\to1$: $\varphi_s\to0$, $\varphi_a\to\pi/8$, and the module reduces to Chen's fast 4-point DCT-II~\cite{ChenSmithFralick}.

\medskip
\noindent\textbf{N = 5.}
\begin{align*}
&\text{Stage 1:}\ \ u_k=\tfrac{x_k+x_{4-k}}{\sqrt2},\ \ v_k=\tfrac{x_k-x_{4-k}}{\sqrt2}\ \ (k=0,1);\\
&\qquad\qquad c=x_2;\\
&\text{Stage 2:}\ \ a_0=\tfrac{u_0+u_1}{\sqrt2},\ \ a_1=\tfrac{u_0-u_1}{\sqrt2},\\
&\qquad\qquad b_0=\tfrac{v_0+v_1}{\sqrt2},\ \ b_1=\tfrac{v_0-v_1}{\sqrt2};\\
&\text{Stage 3:}\ \ (r_0,r_1)=G(\psi)\,(a_0,a_1),\\
&\qquad\qquad (y_1,y_3)=G(\psi)\,(b_0,b_1); \quad \tan2\psi=\rho/2;\\
&\text{Stage 4:}\ \ (y_0,y_2,y_4)^{\mathsf T}=\mathbf{A}_3\,(r_0,r_1,c)^{\mathsf T},
\end{align*}
where the rows of the orthogonal matrix $\mathbf{A}_3$ are, for $i=1,2,3$, the normalized vectors
$\bigl(\tfrac{\sqrt2\rho\,q_1}{\delta_1-\nu_i},\ \tfrac{\sqrt2\rho\,q_2}{\delta_2-\nu_i},\ 1\bigr)$,
with
\[
q_1=\tfrac{\cos\psi-\sin\psi}{\sqrt2},\qquad q_2=-\tfrac{\cos\psi+\sin\psi}{\sqrt2},
\]
\[
\delta_{1,2}=1+\tfrac{\rho^2}{2}\mp\sqrt{\rho^2+\tfrac{\rho^4}{4}},
\]
and $\nu_1<\nu_2<\nu_3$ the roots of the cubic equation
\[
(1-\nu)(1+\rho^2-\nu)^2=\rho^2\bigl[2(1-\nu)+(1+\rho^2-\nu)\bigr].
\]
Note that a single rotation angle $\psi$ serves both branches: this is the residual KLT $\V_2$ of Theorem~\ref{thm:odd}, applied once to the sums and once to the differences.

\medskip
\noindent\textbf{N = 6.}
\begin{align*}
&\text{Stage 1:}\ \ u_k=\tfrac{x_k+x_{5-k}}{\sqrt2},\ \ v_k=\tfrac{x_k-x_{5-k}}{\sqrt2}\ \ (k=0,1,2);\\
&\text{Stage 2 (inner 3-point KLT, both branches):}\\
&\qquad a=\tfrac{u_0+u_2}{\sqrt2},\ \ g_1=\tfrac{u_0-u_2}{\sqrt2},\\
&\qquad (g_0,g_2)=G(\varphi_3)\,(a,\,u_1);\\
&\qquad b=\tfrac{v_0+v_2}{\sqrt2},\ \ h_1=\tfrac{v_0-v_2}{\sqrt2},\\
&\qquad (h_0,h_2)=G(\varphi_3)\,(b,\,v_1); \quad \tan2\varphi_3=2\sqrt2/\rho;\\
&\text{Stage 3:}\ \ (y_0,y_2,y_4)^{\mathsf T}=\Qh^{(s)}(g_0,g_1,g_2)^{\mathsf T},\\
&\qquad\qquad (y_1,y_3,y_5)^{\mathsf T}=\Qh^{(a)}(h_0,h_1,h_2)^{\mathsf T},
\end{align*}
where the rows of $\Qh^{(s)}$ (resp.\ $\Qh^{(a)}$) are the normalized vectors
$\bigl(\tfrac{p_1}{\mu_1-\nu},\tfrac{p_2}{\mu_2-\nu},\tfrac{p_3}{\mu_3-\nu}\bigr)$
over the three roots $\nu$, in ascending order, of the cubic equation
\[
(1-\nu)\bigl[(1+\rho^2-\nu)(\tau-\nu)-\rho^2\bigr]=\rho^2(\tau-\nu),
\]
with $\tau=\tau_s\triangleq1-\rho+\rho^2$ (resp.\ $\tau=\tau_a\triangleq1+\rho+\rho^2$), and
\[
\mu_{1,3}=1+\tfrac{\rho^2}{2}\mp\sqrt{2\rho^2+\tfrac{\rho^4}{4}},\qquad \mu_2=1,
\]
\[
\pb=\Bigl(\tfrac{\cos\varphi_3}{\sqrt2},\ -\tfrac{1}{\sqrt2},\ -\tfrac{\sin\varphi_3}{\sqrt2}\Bigr).
\]

\medskip
\noindent\textbf{N = 7.}
\begin{align*}
&\text{Stage 1:}\ \ u_k=\tfrac{x_k+x_{6-k}}{\sqrt2},\ \ v_k=\tfrac{x_k-x_{6-k}}{\sqrt2}\ \ (k=0,1,2);\\
&\qquad\qquad c=x_3;\\
&\text{Stage 2 (inner 3-point KLT, both branches;}\\
&\qquad\text{$\varphi_3$, $\mu_m$, $\pb$ as for $N=6$):}\\
&\qquad a=\tfrac{u_0+u_2}{\sqrt2},\ \ g_1=\tfrac{u_0-u_2}{\sqrt2},\\
&\qquad (g_0,g_2)=G(\varphi_3)\,(a,\,u_1);\\
&\qquad b=\tfrac{v_0+v_2}{\sqrt2},\ \ h_1=\tfrac{v_0-v_2}{\sqrt2},\\
&\qquad (h_0,h_2)=G(\varphi_3)\,(b,\,v_1);\\
&\text{Stage 3 (residual correction, both branches):}\\
&\qquad (r_0,r_1,r_2)^{\mathsf T}=\Qh^{(\delta)}(g_0,g_1,g_2)^{\mathsf T},\\
&\qquad (y_1,y_3,y_5)^{\mathsf T}=\Qh^{(\delta)}(h_0,h_1,h_2)^{\mathsf T},\\
&\text{Stage 4:}\ \ (y_0,y_2,y_4,y_6)^{\mathsf T}=\mathbf{A}_4\,(r_0,r_1,r_2,c)^{\mathsf T},
\end{align*}
where the rows of $\Qh^{(\delta)}$ are the normalized vectors
$\bigl(\tfrac{p_1}{\mu_1-\delta},\tfrac{p_2}{\mu_2-\delta},\tfrac{p_3}{\mu_3-\delta}\bigr)$
over the three roots $\delta_1<\delta_2<\delta_3$ of the residual cubic equation
\[
\chi(\delta)\triangleq(1-\delta)\bigl[(1+\rho^2-\delta)^2-\rho^2\bigr]-\rho^2(1+\rho^2-\delta)=0,
\]
the coupling vector is $\qb=\Qh^{(\delta)}\pb$, and the rows of the orthogonal matrix $\mathbf{A}_4$ are, for $i=1,\dots,4$, the normalized vectors
$\bigl(\tfrac{\sqrt2\rho\,q_1}{\delta_1-\nu_i},\ \tfrac{\sqrt2\rho\,q_2}{\delta_2-\nu_i},\ \tfrac{\sqrt2\rho\,q_3}{\delta_3-\nu_i},\ 1\bigr)$,
with $\nu_1<\dots<\nu_4$ the roots of the quartic equation
\[
\chi(\nu)\,(1+\rho^2-\nu)=2\rho^2\bigl[(1-\nu)(1+\rho^2-\nu)-\rho^2\bigr].
\]
As in $N=5$, one operator---here the residual KLT $\V_3=\Qh^{(\delta)}\W_3$ of Stage 2--3---serves both branches, and only the arrowhead stage distinguishes them.

\medskip
\noindent\textbf{N = 8.}
\begin{align*}
&\text{Stage 1:}\ \ u_k=\tfrac{x_k+x_{7-k}}{\sqrt2},\ \ v_k=\tfrac{x_k-x_{7-k}}{\sqrt2}\ \ (k=0,\dots,3);\\
&\text{Stage 2 (inner 4-point KLT, both branches;}\\
&\qquad\text{$\varphi_s$, $\varphi_a$ as for $N=4$):}\\
&\qquad a_0=\tfrac{u_0+u_3}{\sqrt2},\ \ a_1=\tfrac{u_1+u_2}{\sqrt2},\\
&\qquad b_0=\tfrac{u_0-u_3}{\sqrt2},\ \ b_1=\tfrac{u_1-u_2}{\sqrt2};\\
&\qquad (g_0,g_2)=G(\varphi_s)\bigl(\tfrac{a_0+a_1}{\sqrt2},\tfrac{a_0-a_1}{\sqrt2}\bigr),\\
&\qquad (g_1,g_3)=G(\varphi_a)\bigl(\tfrac{b_0+b_1}{\sqrt2},\tfrac{b_0-b_1}{\sqrt2}\bigr);\\
&\qquad \text{likewise } (v_0,\dots,v_3)\mapsto(h_0,\dots,h_3);\\
&\text{Stage 3:}\ \ (y_0,y_2,y_4,y_6)^{\mathsf T}=\Qh^{(s)}(g_0,\dots,g_3)^{\mathsf T},\\
&\qquad\qquad (y_1,y_3,y_5,y_7)^{\mathsf T}=\Qh^{(a)}(h_0,\dots,h_3)^{\mathsf T},
\end{align*}
where the rows of $\Qh^{(s)}$ (resp.\ $\Qh^{(a)}$) are the normalized vectors
$\bigl(\tfrac{p_1}{\mu_1-\nu},\dots,\tfrac{p_4}{\mu_4-\nu}\bigr)$
over the four ascending roots $\nu$ of the quartic equation
\[
\chi_s(\nu)\,\chi_a(\nu)=\rho(1-\rho)\,\chi(\nu)
\]
(resp.\ $\chi_s(\nu)\,\chi_a(\nu)=-\rho(1+\rho)\,\chi(\nu)$), with the parity quadratic polynomials
\[
\chi_s(\nu)\triangleq(1-\nu)(\tau_s-\nu)-\rho^2,\quad
\chi_a(\nu)\triangleq(1-\nu)(\tau_a-\nu)-\rho^2
\]
($\tau_s,\tau_a$ as for $N=6$), $\chi$ the residual cubic polynomial of the $N=7$ module, and
\[
\mu_{1,3}=1-\tfrac{\rho(1-\rho)}{2}\mp\tfrac{\rho}{2}\sqrt{(1-\rho)^2+4},
\]
\[
\mu_{2,4}=1+\tfrac{\rho(1+\rho)}{2}\mp\tfrac{\rho}{2}\sqrt{(1+\rho)^2+4},
\]
\[
p_1=\tfrac{\cos\varphi_s+\sin\varphi_s}{2},\qquad
p_2=-\tfrac{\cos\varphi_a+\sin\varphi_a}{2},
\]
\[
p_3=\tfrac{\cos\varphi_s-\sin\varphi_s}{2},\qquad
p_4=\tfrac{\sin\varphi_a-\cos\varphi_a}{2}.
\]

The two branch quartic polynomials differ only in the coefficient multiplying $\chi$---namely $-\rho(1-\rho)$ and $+\rho(1+\rho)$, the two rank-one strengths of the even fold.

\medskip
All seven modules have been verified against direct eigendecomposition to machine precision across $\rho\in(0,1)$; a C reference implementation of the modules and of the verification procedure is available in the accompanying open-source repository~\cite{AR1KLTrepo}. The cubic and quartic equations above are the characteristic equations of the corresponding parity blocks of $\T_N$, written in the variables of Sections~\ref{sec:even}--\ref{sec:odd}. Their roots interlace as stated there, which fixes the output ordering without further computation.

\section{Fast application and complexity}
\label{sec:complexity}

The only dense factors in~\eqref{eq:even-thm} and~\eqref{eq:odd-thm} are the correction stages $\Qh$ and the arrowhead diagonalizer $\A_{M+1}$, and none of them is unstructured. By~\eqref{eq:secular-vectors}, the rows of each $\Qh$ are, up to row-wise normalization, rows of the Cauchy matrix $[1/(\mu_j-\nu_i)]$. Absorbing the entries of $\pb$ into a column scaling, $\Qh=\Db_1\Cb\,\Db_2$ with $\Cb$ Cauchy and $\Db_1,\Db_2$ diagonal. The same holds for $\A_{M+1}$ by~\eqref{eq:arrowhead-vectors}, with one appended column. Multiplication of a vector by a Cauchy matrix is the classical Trummer problem: it can be performed exactly in $O(M\log^2M)$ operations~\cite{Gerasoulis}, and to any fixed accuracy in $O(M)$ operations by the fast multipole method~\cite{GreengardRokhlin}. This is precisely the technology of the fast divide-and-conquer symmetric eigensolvers~\cite{Cuppen,GuEisenstat}, whose diagonal-plus-rank-one update structure the present factorizations share; the interlacing of poles and roots (Section~\ref{sec:secular}) provides the separation on which these methods rely.

The resulting complexity is easiest to state in the even case. For $N=2^t$, Theorem~\ref{thm:even} gives the cost recursion
\[
C(N)=2\,C(N/2)+2\,A(N/2)+O(N),
\]
where $A(M)$ is the cost of one correction stage and the $O(N)$ term covers the butterfly. Applied densely, $A(M)=\Theta(M^2)$ and $C(N)=\Theta(N^2)$: the recursion then offers no asymptotic gain over direct application of $\W_N$, although the sum-branch correction, of strength $\sigma_s=\rho(1-\rho)$, weakens at high correlation and can be truncated in practice. With exact Cauchy multiplication, $A(M)=O(M\log^2M)$ and $C(N)=O(N\log^3N)$. With fast-multipole application at fixed accuracy, $A(M)=O(M)$ and
\[
C(N)=O(N\log N),
\]
the same order as the fast DCT, to which the recursion degenerates at $\rho\to1$.

The odd-order recursion has the same profile. By~\eqref{eq:V}, the residual KLT costs one Cauchy-structured correction beyond $\W_M$, and the arrowhead stage is a Cauchy matrix with one extra column. Mixed-parity recursions therefore also run in $O(N\log N)$ with fast-multipole-accelerated stages. The repository~\cite{AR1KLTrepo} includes a reference implementation of both recursions (dense stages, one shared half-size plan per level), verified against direct eigendecomposition for orders up to $1000$.

\section{Discussion and conclusions}
\label{sec:conclusions}

The question this paper revisited is the one that motivated the invention of the DCT: how to compute the optimal transform of an AR(1) source. The answer given in 1974 was an approximation with a fast algorithm. The answer given here is that the exact transform has a fast structure of its own, at every order. For even $N$, the KLT folds into two copies of itself at half order, corrected by rotations that rank-one boundary perturbations of strengths $\rho(1\mp\rho)$ generate. For odd $N$, it folds into two copies of the residual KLT---the KLT of the block conditioned on a boundary sample---with the center sample entering through an arrowhead border. The residual KLT recurses through half order as well (Corollary~\ref{cor:resfam}). Both factorizations follow from the covariance by the Sherman--Morrison identity and the classical secular eigen-updates. At $\rho\to1$ they degenerate to the classical DCT-II/DCT-IV split underlying Chen's fast DCT, and to the DCT-VI/DST-VII split of the odd-length DCT-II. And with fast-multipole application of the Cauchy-structured correction stages, both recursions run in $O(N\log N)$ operations (Section~\ref{sec:complexity}): the exact KLT admits the same complexity order as the FFT and the fast trigonometric transforms.

The organizing principle of both factorizations is the phase equation's content---boundary reflection phases: the entire deviation of the AR(1) KLT from a translation-invariant fast transform lives in boundary terms, and boundary terms are rank-one perturbations or, at odd order, rank-one borders.

The construction also clarifies the classical landscape. The asymmetry of Chen's split---the DCT-II pairs with the DCT-IV rather than with itself---persists at the exact-KLT level at every $\rho$, quantified by the strengths $\rho(1-\rho)$ and $\rho(1+\rho)$. The odd DCT-II splits into two \emph{different} transforms because, conditioned on the center sample, a single residual KLT drives both branches. The DST-VII---adopted for intra-prediction residuals in modern video coding~\cite{HanSaxenaRose,ReznikICASSP}---is exhibited as the $\rho\to1$ face of an object that exists exactly at every $\rho$; the same limit yields the conversion identity of Corollary~\ref{cor:conversion}.

Several directions follow. On the practical side: implementations that exploit the structure fully, including fast construction and fast-multipole application of the correction stages, and a fresh look at compression applications where $\rho$ is known or can be estimated---image and video coding foremost among them. On the theoretical side: extending the construction to richer sources, beginning with AR($p$), where one expects rank-$p$ boundary terms and wider borders.

\end{document}